\documentclass[conference]{IEEEtran}
\IEEEoverridecommandlockouts
\usepackage{hyperref}
 \usepackage{lipsum,graphicx,multicol}
 \usepackage{comment}
\usepackage{color,soul}
\usepackage{mathtools}
\usepackage{graphicx}
\usepackage[top=0.75in, bottom=1.05 in, left=0.64in, right=0.64in]{geometry}
\usepackage{color,soul}
\usepackage{setspace}
\usepackage{multirow}
\usepackage{tablefootnote}
\usepackage[norelsize, linesnumbered, ruled, lined, boxed, commentsnumbered]{algorithm2e}
\usepackage[normalem]{ulem}
\usepackage{comment}
\usepackage[font=small,labelfont=bf]{caption}
\usepackage{subcaption}
\usepackage{bm, amsmath, amssymb, amsthm}
\usepackage{amsfonts}
\usepackage {amssymb,amsmath,xcolor}
\usepackage{booktabs}
\usepackage{cite}
\usepackage{array}
\usepackage[acronym]{glossaries}
\usepackage{cases}
\usepackage{cleveref}
\usepackage{algorithmic}

\usepackage{amsthm}
\usepackage{array}
\def\BibTeX{{\rm B\kern-.05em{\sc i\kern-.025em b}\kern-.08em
    T\kern-.1667em\lower.7ex\hbox{E}\kern-.125emX}}
    
  \newtheorem{theorem}{Theorem}%[section]

\newtheorem{problem}{Problem}
\newtheorem{corollary}{Corollary}
\newtheorem{proposition}{Proposition}
\newtheorem{lemma}{Lemma}
\newtheorem{remark}{Remark}

\newcommand{\bieee}{\begin{IEEEeqnarray}{rCl}}
\newcommand{\eieee}{\end{IEEEeqnarray}}
\usepackage{amsthm,mdframed,calc}
 \usepackage{float}
\usepackage{diagbox}
\SetKwInput{KwInput}{Input}                % Set the Input
\SetKwInput{KwOutput}{Output}              % set the Output
\title{Basis Selection For Analog Coded Computing Under Differential Privacy Constraints}

\author{\IEEEauthorblockN{Aaditya Sharma, Rimpi Borah, and J. Harshan}
\IEEEauthorblockA{Department of Electrical Engineering, Indian Institute of Technology Delhi, India.}}

\begin{document}
\maketitle
\begin{abstract}
In analog coded computing, encoding basis functions are typically chosen to guarantee exact reconstruction and improve numerical stability. However, the same basis functions also determine the statistical geometry of the observations available to colluding workers, and hence directly influence the privacy of the underlying data. This suggests that basis selection should not be viewed solely as a reconstruction or numerical conditioning problem, but rather as a joint privacy and accuracy co-design problem. Motivated by this observation, we revisit Analog Lagrange Coded Computing (ALCC) and show that the full orthogonality constraints imposed by the conventional Lagrange basis are stronger than what is required for exact recovery. We retain only the interpolation constraints necessary for the data bearing basis functions and require the noise bearing basis functions merely to vanish at the data recovery points. This yields a parameterized family of generalized polynomial encoding bases that preserves exact reconstruction while introducing additional design degrees of freedom. We then derive a basis-dependent characterization of differential privacy for colluding workers and formulate a privacy-calibrated numerical optimization problem that jointly selects the data and noise basis functions to reduce reconstruction error while satisfying a prescribed \((\epsilon,\delta)\)-DP requirement. The resulting framework therefore treats the encoding basis itself as a privacy-design variable, rather than selecting a fixed closed-form basis a priori. The same design philosophy is also compatible with alternative privacy metrics such as mutual-information leakage. Numerical results show that, at the same privacy level, optimized basis functions substantially reduce floating-point reconstruction error relative to conventional ALCC.

\end{abstract}
\section{Introduction}
\noindent
Distributed coded computing has emerged as an effective framework for
large scale computation, where a master node distributes coded shares of
data among multiple workers and reconstructs the desired computation from
their returned results
\cite{lcc,b1,analogPrivacy,polycodes,b8,b2}.
In the analog setting, two key challenges arise:
\emph{finite-precision reconstruction error} and
\emph{privacy leakage} to honest-but-curious or colluding workers.
Existing works address these issues through coding structure,
interpolation-node design, and artificial noise injection. For instance,
Chebyshev nodes improve numerical stability \cite{b2}, while ALCC
\cite{b1} characterizes privacy through mutual information security.
These observations highlight the privacy-accuracy tradeoff inherent in
analog coded computing.

A key observation of this work is that the \emph{encoding basis itself}
influences both sides of this tradeoff. It affects the numerical
amplification experienced during distributed computation while also
shaping the statistical geometry of the observations available to
colluding workers. 
%Hence, two schemes that are equivalent from the viewpoint of exact recoverability need not be equivalent in either numerical accuracy or privacy. 
Hence, rather than fixing a basis and
subsequently analyzing its performance (which is usually what is done till now), we therefore treat the basis
itself as a \emph{joint privacy--accuracy design variable}.
\begin{figure}[t]
    \centering
 \includegraphics[width=0.75\columnwidth]{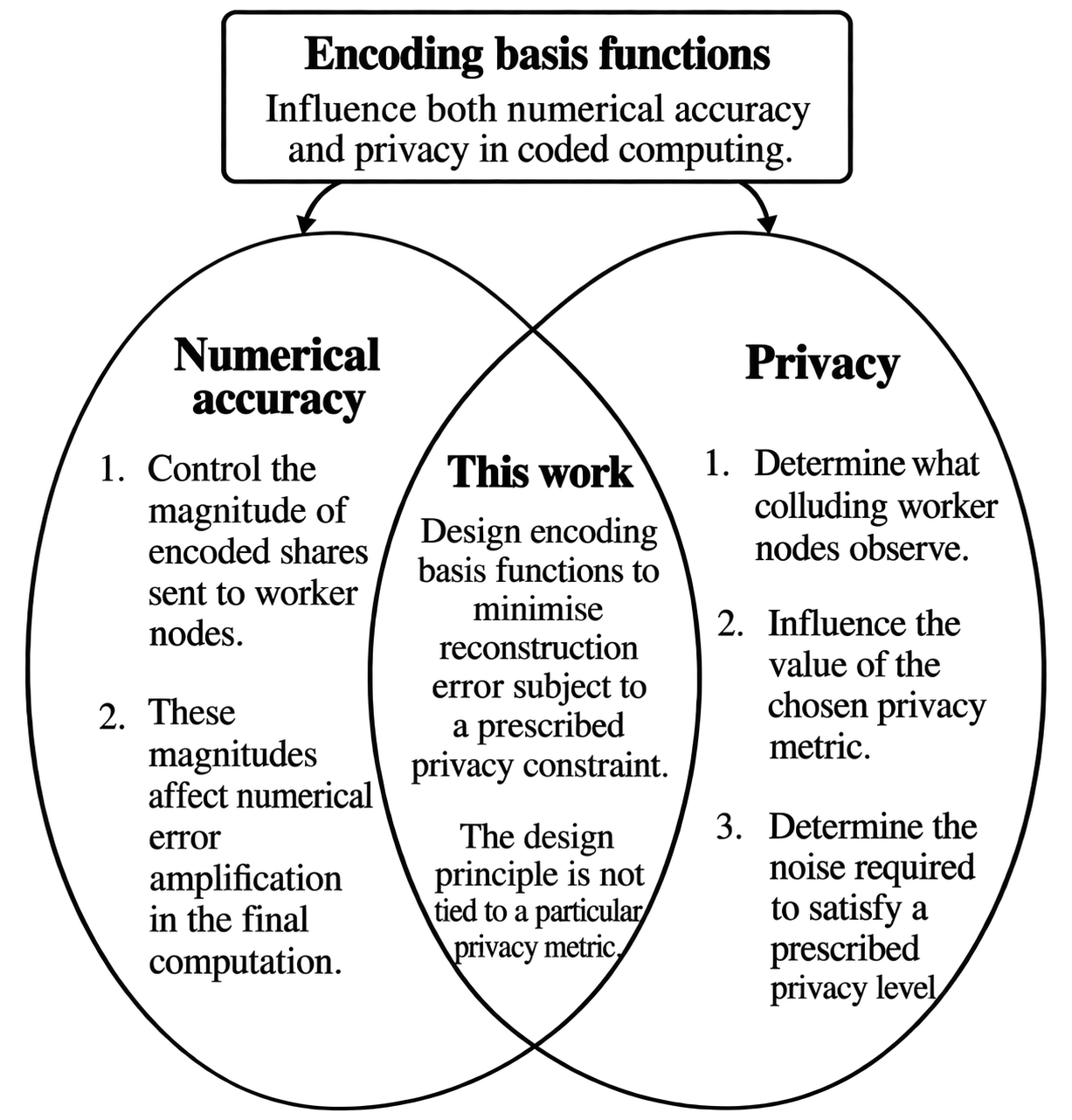}
    \vspace{-1.5mm}
    \caption{Privacy--accuracy co-design of the encoding basis.}
    \label{fig:privacyAccuracyVenn}
    \vspace{-0.4cm}
\end{figure}
As illustrated in Fig.~\ref{fig:privacyAccuracyVenn}, the encoding basis
is the common design object underlying both objectives: its numerical
behavior influences reconstruction accuracy, while the encoded shares
it produces determine the information exposed to colluding workers.
Their intersection therefore represents the central premise of this
work, i.e., the encoding basis should be designed for accuracy and privacy
jointly, rather than for recoverability alone.

We show that exact reconstruction does not require the full interpolation structure imposed by the conventional Lagrange basis, leaving exploitable degrees of freedom in the encoding design. In particular, orthogonality should only be maintained across the data bearing basis functions, and the noise bearing basis functions merely need to vanish at the encoding points. Therefore, we can parameterize the resulting family of recovery-preserving polynomial bases, characterize their effect on $(\epsilon,\delta)$-differential privacy and finite-precision reconstruction accuracy, and optimize the basis under a prescribed privacy requirement. Thus, privacy does not merely constrain the code after its selection, it participates in the design of the code itself. Numerical results show that, at the same privacy level, the optimized bases achieve substantially lower floating-point reconstruction error than conventional Lagrange-based ALCC.

\section{Background AND Proposed Framework}
\label{sec:proposed framwork}

\subsection{Preliminaries on the ALCC Framework}
We consider the ALCC-based distributed computing framework \cite{b1}
comprising a master server connected to $N$ worker nodes through
dedicated communication links such that the master node communicates
directly with each worker, while there is no direct communication among
the workers themselves. Let the dataset held by the master node be
$\mathbf{X}=[\mathbf{X}_1,\mathbf{X}_2,\dots,\mathbf{X}_K]$, where
$\mathbf{X}_j\in\mathbb{R}^{m\times n}$ for all $j\in[K]$. The entries
of the dataset are independently sampled from an unknown probability
distribution with bounded support over the interval $[-r,r]$, for some
$r\in\mathbb{R}^{+}$. The objective of the framework is to compute a polynomial function
$f:\mathbb{R}^{m\times n}\rightarrow\mathbb{R}^{u\times h}$ of degree
at most $D$ over the dataset by leveraging the computational power of
$N$ distributed workers. The function $f(\mathbf{X})$ is said to be a
degree-$D$ polynomial function if each entry of $f(\mathbf{X})$ is a
multivariate polynomial in the entries of $\mathbf{X}$ with total degree
at most $D$. Furthermore, up to $t<N$ workers may collude and attempt to
infer information about the dataset from their received shares. To distribute the shares of the dataset among the workers, the ALCC framework in \cite{b1} employs a secret-sharing-based encoding strategy
\cite{b5}. Specifically, the master constructs the encoding polynomial

\begin{equation}
\mathbf{u}(z)
=
\sum_{j=1}^{K}\ell_j(z)\mathbf{X}_j
+
\sum_{r=1}^{t}\ell_{K+r}(z)\mathbf{N}_r,
\label{eq:alcc_encoding}
\end{equation}

\noindent where $\ell_j(z)$ denotes the $j$th Lagrange basis polynomial, defined as
$\ell_j(z)=
\prod_{\substack{\ell\in[K+t],\,\ell\neq j}}
\frac{z-\beta_\ell}{\beta_j-\beta_\ell}$.
The corresponding interpolation points are chosen as scaled roots of
unity,
$\beta_k=\beta\exp\!\left(\frac{2\pi i(k-1)}{K+t}\right)$,
for $k\in[K+t]$, where $\beta\in\mathbb{R}^{+}$ is a scaling
parameter. The matrices $\{\mathbf{N}_r\}_{r=1}^{t}$ in \eqref{eq:alcc_encoding} denote the noise matrices introduced to protect the underlying dataset. Their entries are independently drawn from a circularly symmetric complex Gaussian distribution with variance $\sigma_n^2/t$, i.e., $[\mathbf{N}_r]_{ab}\sim \mathcal{CN}(0,\sigma_n^2/t)$. Once the encoding polynomial $\mathbf{u}(z)$ is constructed, the master
evaluates it at the worker evaluation points
$\alpha_i=\exp\!\left(\frac{2\pi i(i-1)}{N}\right)$,
$i\in[N]$, which are chosen as the $N$th roots of unity. The $i$th
worker receives the encoded share $\mathbf{u}(\alpha_i)$, computes
$f(\mathbf{u}(\alpha_i))$, and returns the resulting computation to the
master node.

Since each Lagrange basis polynomial in \eqref{eq:alcc_encoding} has degree at most $K+t-1$, the encoding polynomial satisfies $\deg(\mathbf{u}(z))\leq K+t-1$. Since $f(\cdot)$ has degree at most $D$, the composition $f(\mathbf{u}(z))$ is therefore a polynomial in $z$ of degree at most $\widetilde{D}=D(K+t-1)$. Consequently, under infinite-precision arithmetic, after receiving the computations of any $\widetilde{D}+1=D(K+t-1)+1$ workers, the master first interpolates the polynomial $f(\mathbf{u}(z))$ from the corresponding evaluations $\{f(\mathbf{u}(\alpha_i))\}$.
After reconstructing $f(\mathbf{u}(z))$, the master evaluates the interpolated polynomial at the first $K$ encoding points $\{\beta_j\}_{j=1}^{K}$. The Lagrange basis polynomials satisfy the orthogonality property $\ell_k(\beta_j)=\delta_{kj}$ for $k\in[K+t]$ and $j\in[K+t]$. Hence, for every $j\in[K]$, $\mathbf{u}(\beta_j)=\mathbf{X}_j$, which immediately gives $f(\mathbf{u}(\beta_j))=f(\mathbf{X}_j)$. Therefore, under infinite-precision arithmetic, the master exactly recovers the desired outputs $\{f(\mathbf{X}_j)\}_{j=1}^{K}$.

\begin{remark}
\label{rem:recovery_freedom}
Conventional ALCC imposes $\ell_k(\beta_j)=\delta_{kj}$ over all $K+t$ interpolation points. However, exact recovery only requires $\phi_k(\beta_j)=\delta_{kj}$ for $k,j\in[K]$ and $\psi_r(\beta_j)=0$ for $r\in[t],\,j\in[K]$. These conditions ensure $\mathbf{u}(\beta_j)=\mathbf{X}_j$ and hence $f(\mathbf{u}(\beta_j))=f(\mathbf{X}_j)$. Thus, the full $(K+t)$-point Lagrange orthogonality structure is stronger than necessary, relaxing the remaining constraints exposes additional basis-design degrees of freedom without altering the recovery condition.
\end{remark}
%\begin{remark}
%\label{rem:minimal_recovery_conditions}
%Conventional ALCC imposes$\ell_k(\beta_j)=\delta_{kj}$ over all $K+t$ interpolation points. However, exact recovery only requires 5$\phi_k(\beta_j)=\delta_{kj}$ for $k,j\in[K]$ and $\psi_r(\beta_j)=0$ for %$r\in[t],\,j\in[K]$. These conditions alone ensure $\mathbf{u}(\beta_j)=\mathbf{X}_j$ and hence $f(\mathbf{u}%(\beta_j))=f(\mathbf{X}_j)$. \end{remark}

%\begin{remark}
%\label{rem:orthogonality_freedom}
%Thus, the full $(K+t)$-point Lagrange orthogonality structure is stronger than necessary for exact reconstruction. In particular, the noise bearing basis functions need only vanish at the $K$ data recovery points. Relaxing the remaining constraints exposes additional basis design degrees of freedom without changing the recovery condition. \end{remark}

\subsection{Parameterized Basis Design}

Motivated by Remark~\ref{rem:recovery_freedom}, we parameterise the ALCC encoding in
\eqref{eq:alcc_encoding} by replacing the fixed Lagrange basis polynomials with arbitrary polynomial basis functions that preserve only the orthogonality conditions required for exact recovery.

In this context, let $\{\phi_j(z)\}_{j=1}^{K}$ denote the basis functions associated with
the data matrices $\{\mathbf{X}_j\}_{j=1}^{K}$, and let
$\{\psi_r(z)\}_{r=1}^{t}$ denote the basis functions associated with the noise matrices $\{\mathbf{N}_r\}_{r=1}^{t}$. We restrict
these basis functions to polynomials of degree at most $K+t-1$ so that
the degree of the encoding polynomial, and hence the recovery threshold,
remain unchanged from conventional ALCC. The generalized encoding
polynomial is given by

\begin{equation}
\mathbf{u}(z)
=
\sum_{j=1}^{K}\phi_j(z)\mathbf{X}_j
+
\sum_{r=1}^{t}\psi_r(z)\mathbf{N}_r.
\label{eq:New encoding}
\end{equation}

To preserve the original reconstruction mechanism, the proposed basis
functions are required to satisfy
$\phi_k(\beta_j)=\delta_{kj}$ for all $k,j\in[K]$ and
$\psi_r(\beta_j)=0$ for all $r\in[t]$ and $j\in[K]$. These conditions
ensure that $\mathbf{u}(\beta_j)=\mathbf{X}_j$ and hence
$f(\mathbf{u}(\beta_j))=f(\mathbf{X}_j)$ for every $j\in[K]$. Note that, unlike conventional ALCC, which employs $K+t$ equally spaced
interpolation points on the circle of radius $\beta$, the generalized
construction requires only $K$ data-recovery points. In this work, these
are chosen as
$\beta_j=\beta\exp\!\left(2\pi\mathrm{i}(j-1)/K\right)$,
$j\in[K]$, i.e., as $K$ equally spaced points on the same circle.

We next characterize the polynomial basis family satisfying these
relaxed orthogonality conditions. Define

\begin{equation}
q(z)=\prod_{\ell=1}^{K}(z-\beta_\ell),
\label{qz}
\end{equation}

which satisfies $q(\beta_j)=0$ for every $j\in[K]$. Let
$\widetilde{\ell}_j(z)$ denote the $j$th Lagrange basis polynomial
constructed only over the $K$ data-recovery points
$\{\beta_\ell\}_{\ell=1}^{K}$, i.e.,
$\widetilde{\ell}_j(z)=
\prod_{\substack{\ell\in[K],\,\ell\neq j}}
\frac{z-\beta_\ell}{\beta_j-\beta_\ell}$.
By construction,
$\widetilde{\ell}_k(\beta_j)=\delta_{kj}$ for all $k,j\in[K]$. The basis functions correspond to the dataset can be parameterized as

\begin{equation}
\phi_j(z)=\widetilde{\ell}_j(z)
+
q(z)\eta_j(z),
\qquad j\in[K],
\label{phi basis1}
\end{equation}

while the basis functions correspond to noise matrices can be parameterized as

\begin{equation}
\psi_r(z)
=
q(z)\xi_r(z),
\qquad r\in[t],
\label{psi basis}
\end{equation}

where
$\eta_j(z),\xi_r(z)\in\mathcal{P}_{t-1}$ and
$\mathcal{P}_{t-1}$ denotes the space of polynomials of degree at most
$t-1$.

The form in \eqref{phi basis1} follows from the fact that
$\phi_j(z)-\widetilde{\ell}_j(z)$ vanishes at every
$\beta_\ell$, $\ell\in[K]$. Therefore,
$q(z)$ must divide
$\phi_j(z)-\widetilde{\ell}_j(z)$, and hence it can be expressed as
$q(z)\eta_j(z)$. Similarly, since every admissible
$\psi_r(z)$ vanishes at all $K$ data-recovery points, $q(z)$ must divide
$\psi_r(z)$, giving the representation in \eqref{psi basis}. Since all
basis functions are restricted to have degree at most $K+t-1$, the quotient polynomials $\eta_j(z)$ and $\xi_r(z)$ must have degree at most $t-1$. Thus, \eqref{phi basis1}--\eqref{psi basis} do not merely describe a
particular alternative to the conventional Lagrange construction.
Rather, they parameterize the recovery-preserving polynomial basis
family considered in this work. The coefficients of
$\{\eta_j(z)\}_{j=1}^{K}$ and $\{\xi_r(z)\}_{r=1}^{t}$ therefore form
the additional design variables through which the numerical accuracy
and privacy properties of the distributed computation can be controlled
while exact recoverability is preserved by construction. Therefore, the remaining stages of the protocol and the recovery
threshold remain identical to conventional ALCC. Thus, under
infinite-precision arithmetic, the proposed construction preserves both
the exact recovery property and the recovery threshold of the original
ALCC framework.

In practical floating-point implementations, however, the distributed
computation and polynomial interpolation introduce numerical errors.
Let $\widehat{f(\mathbf{X}_j)}$ denote the reconstructed output obtained
using the proposed framework, while $f(\mathbf{X}_j)$ denotes the
corresponding exact centralized computation. Therefore, we next characterize the finite-precision error induced by the generalized basis functions in
\eqref{eq:New encoding}. We then show that the same basis functions also
determine the statistical observations available to colluding workers
and, consequently, the privacy guarantees of the framework. This common
dependence on the encoding basis forms the foundation of the
privacy--accuracy co-design developed in the remainder of the paper.

%%%%%%%%%%%%%%%%%%%%%%%%%%%%%%%%%%%%%%%%%%%%%%%%%%%%%%%%%%%%%%%%%%%%%%%%%%%%%%%%%%%%%%%%%%
\section{Accuracy and Privacy Analysis under Generalized Basis}
\label{sec:accuracy and privacy}
In this section, we analyze the accuracy and privacy guarantees of the proposed framework under generalized basis functions defined in \eqref{phi basis1} and \eqref{psi basis}. 
%%%%%%%%%%%%%%%%%%%%%%%%%%%%%%%%%%%%%%%%%%%%%%%%%%%%%%%%%%%%%%%%%%%%%%%%%%%%%%
\subsection{Accuracy Characterization}
\label{sec:Accuracy analysis}

We characterize the finite-precision reconstruction error under the generalized basis functions in \eqref{phi basis1} and \eqref{psi basis} by carrying out similar analysis in \cite[Theorem 4]{b1}. Recall that the entries of the artificial noise matrices are independently drawn according to
$(\mathbf N_r)_{pq}
    \sim
    \mathcal{CN}
    \left(
        0,\frac{\sigma_n^2}{t}
    \right), r\in[t].$
For $\theta>0$, define the event
$\mathcal E_\theta =
\left\{\max_{r\in[t]}\|\mathbf N_r\|_{\max}
\leq\frac{\theta\sigma_n}{\sqrt t}\right\}.$
Further, define

\begin{equation}
\label{eq:M_a}
M=\max_{i\in[N]}\left(r\sum_{j=1}^{K}|\phi_j(\alpha_i)|+\frac{\theta\sigma_n}{\sqrt t}
\sum_{r=1}^{t}|\psi_r(\alpha_i)|\right).
\end{equation}

Let $\widetilde D=(K+t-1)D$, and let
$\mathcal S=\{i_1,\ldots,i_{\widetilde D+1}\}\subseteq[N]$
denote the set of workers used for reconstruction. We denote the
corresponding Vandermonde matrix by
$\mathcal V\in\mathbb C^{(\widetilde D+1)\times(\widetilde D+1)}$,
whose entries are $[\mathcal V]_{a,\ell+1}
    =
    \alpha_{i_a}^{\ell},
    \
    a\in[\widetilde D+1],\quad
    \ell=0,\ldots,\widetilde D.$
Also, define
$\Theta_\beta=
\frac{\beta^{2(\widetilde D+1)}-1}{\beta^2-1}.$
For each output entry $(p,q)$, let $\mathbf a_{pq}$ and
$\widehat{\mathbf a}_{pq}$ denote the exact and computed coefficient
vectors, respectively, of the polynomial reconstructed from the worker
evaluations. As in the finite precision interpolation model considered
in \cite{b1}, we assume
\begin{equation}
    \|
        \widehat{\mathbf a}_{pq}
        -
        \mathbf a_{pq}
    \|_2
    \leq
    \kappa_2(\mathcal V)
    \|\mathbf a_{pq}\|_2
    2^{-b_m},
    \label{eq:finite_precision_interpolation}
\end{equation}
where $b_m$ denotes the floating-point precision in bits.

\begin{theorem}
\label{th:accuracy}
Let \(f:\mathbb{R}^{m\times n}\rightarrow\mathbb{R}^{u\times h}\)
be a polynomial function of degree \(D\), and assume that the magnitude
of its coefficients is bounded by $c$. Further, assume that
$\|\mathbf X_j\|_{\max}\le r$ for all \(j\in[K]\), and
$(emn)M\geq1$. Then, for all $j\in[K]$, $p\in[u]$, and $q\in[h]$
such that $f(\mathbf X_j)_{pq}\neq0$,

\begin{equation}
\label{accuracy bound}
\begin{aligned}
&\frac{\bigl|\widehat{f(\mathbf X_j)}_{pq}
-
f(\mathbf X_j)_{pq}\bigr|}
{\bigl|f(\mathbf X_j)_{pq}\bigr|} \\
&\quad \leq
\frac{\sqrt{\Theta_\beta}\,
\kappa_2(\mathcal V)\,
\sqrt{\widetilde D+1}}
{\sigma_{\min}(\mathcal V)\,
\bigl|f(\mathbf X_j)_{pq}\bigr|}
\\
&\qquad \times
c(D+1)(emn)^D M^D\,2^{-b_m}.
\end{aligned}
\end{equation}

with probability at least
$\left(1-e^{-\theta^2}\right)^{tmn}
\geq 1-tmn e^{-\theta^2}$.
\end{theorem}

\begin{proof}
    See Appendix~\eqref{proof:theorem1}.
\end{proof}
\begin{remark}
\label{rem:M_accuracy}
Since the relative-error bound in \eqref{accuracy bound} scales as
$M^D$, reducing $M$ improves the basis-dependent component of the
accuracy bound. Hence, $M$ serves as the accuracy criterion for the
basis-design problem.
\end{remark}

The entrywise relative error can become large when
$\bigl|f(\mathbf X_j)_{pq}\bigr|$ is close to zero, even when the
absolute reconstruction error remains small. Therefore, in the numerical
evaluation we consider the normalized Frobenius error
$e_{\mathrm F}^{(j)}
=
\frac{
\|\widehat{f(\mathbf X_j)}-f(\mathbf X_j)\|_F
}{
\|f(\mathbf X_j)\|_F
}$,
for $\|f(\mathbf X_j)\|_F>0$, which is less sensitive to individual
near-zero output entries.

\begin{lemma}[Choice of $\theta$]
\label{lem:choice_theta}
For any prescribed failure probability $\zeta\in[0,1]$, and for given
system parameters $m$, $n$, and $t$, the minimum value of $\theta$ for
which the reconstruction error bound in \eqref{accuracy bound} holds
with probability at least $1-\zeta$ is
\begin{equation}
\theta_{\min}
=
\sqrt{
-\log\left(
1-(1-\zeta)^{1/(tmn)}
\right)
}.
\label{eq:theta_min}
\end{equation}
Consequently, any $\theta\geq\theta_{\min}$ guarantees that
\eqref{accuracy bound} holds with probability at least $1-\zeta$.
\end{lemma}
\begin{proof}
    See Appendix\eqref{proof:lemma1}
\end{proof}

\begin{remark}
\label{rem:theta_interpretation}
For fixed system parameters $(m,n,t)$, Lemma~\ref{lem:choice_theta}
allows $\theta$ to be selected for a confidence level.
Choosing $\theta\geq\theta_{\min}$ ensures that the reconstruction-error
bound holds with probability at least $1-\zeta$. Increasing $\theta$
improves this confidence, but increases $M$ in \eqref{eq:M_a},
thereby loosening the error bound.
\end{remark}

\subsection{Privacy Characterization}
\label{sec:privacy analysis}
We next characterize the privacy guarantees of the generalized framework
using differential privacy (DP) \cite{b3,b3new,b9} and mutual information security (MIS).
A coalition of workers observes linear combinations of the data and
artificial noise, whose coefficients are determined by the generalized
basis functions. To make this dependence explicit, write
$\eta_j(z)=\sum_{n=0}^{t-1}a_{j,n}z^n$ and
$\xi_r(z)=\sum_{n=0}^{t-1}b_{r,n}z^n$, and collect their coefficients as
$\mathbf{A}=[a_{j,n}]\in\mathbb C^{K\times t}$ and
$\mathbf{B}=[b_{r,n}]\in\mathbb C^{t\times t}$.

Let $\mathcal T=\{i_1,\ldots,i_t\}\subset[N]$ denote a coalition of
$t$ workers. Define
$\mathbf x=\mathrm{vec}(\mathbf X_1,\ldots,\mathbf X_K)$,
$\mathbf n=\mathrm{vec}(\mathbf N_1,\ldots,\mathbf N_t)$, and
$\mathbf y_{\mathcal T}
=\mathrm{vec}(\mathbf Y_{i_1},\ldots,\mathbf Y_{i_t})$,
where $\mathbf Y_i=\mathbf u(\alpha_i)$. Further, define the basis
evaluation matrices

\begin{equation}
\label{PhiPsiMatrix}
\begin{aligned}
\mathbf{\Phi}_{\mathcal{T}}(\mathbf{A}) &=
\begin{bmatrix}
\phi_1(\alpha_{i_1}) & \cdots & \phi_K(\alpha_{i_1}) \\
\vdots & \ddots & \vdots \\
\phi_1(\alpha_{i_t}) & \cdots & \phi_K(\alpha_{i_t})
\end{bmatrix},
\\[2em]
\mathbf{\Psi}_{\mathcal{T}}(\mathbf{B}) &=
\begin{bmatrix}
\psi_1(\alpha_{i_1}) & \cdots & \psi_t(\alpha_{i_1}) \\
\vdots & \ddots & \vdots \\
\psi_1(\alpha_{i_t}) & \cdots & \psi_t(\alpha_{i_t})
\end{bmatrix}.
\end{aligned}
\end{equation}

The observations available to the coalition can therefore be written as

\begin{equation}
\label{coalition model}
\mathbf y_{\mathcal T}
=
(\mathbf\Phi_{\mathcal T}(\mathbf A)\otimes\mathbf I_{mn})\mathbf x
+
(\mathbf\Psi_{\mathcal T}(\mathbf B)\otimes\mathbf I_{mn})\mathbf n.
\end{equation}

Thus, the colluding workers observe a linear Gaussian mechanism whose
statistical geometry depends explicitly on the selected basis functions. For DP, we use entry-level adjacency: two datasets $\mathbf X\sim\mathbf X'$ are neighboring if they differ in at most one
scalar entry. Since all entries lie in $[-r,r]$, their maximum change is
$\Delta=2r$. The corresponding $\ell_2$-sensitivity of the coalition
mean is

\begin{equation}
\label{Delta_M}
\Delta_{\mathcal M,\mathcal T}
=
2r\max_{j\in[K]}
\|\mathbf\Phi_{\mathcal T}(\mathbf{A})(:,j)\|_2.
\end{equation}

Recall that a randomized mechanism $\mathcal M$ satisfies
$(\epsilon,\delta)$-DP if, for every $\mathbf X\sim\mathbf X'$ and
measurable set $S$,
$\mathbb P(\mathcal M(\mathbf X)\in S)
\leq e^\epsilon\mathbb P(\mathcal M(\mathbf X')\in S)+\delta$. Assuming $\mathbf\Psi_{\mathcal T}(\mathbf{B})$ is nonsingular (it can be shown that assuming non-singularity of B ensures non-singularity of the Gram matrix appearing in the expression below), define

\begin{equation}
\label{eq:G_T}
\mathbf G_{\mathcal T}(\mathbf{A},\mathbf{B})
=
\mathbf\Phi_{\mathcal T}(\mathbf{A})^H
\left(
\mathbf\Psi_{\mathcal T}(\mathbf{B})
\mathbf\Psi_{\mathcal T}(\mathbf{B})^H
\right)^{-1}
\mathbf\Phi_{\mathcal T}(\mathbf{A}),
\end{equation}
and
\[
\Gamma_t(\mathbf{A},\mathbf{B})
=
\max_{\substack{\mathcal T\subset[N]\\|\mathcal T|=t}}
\max_{j\in[K]}
[\mathbf G_{\mathcal T}(\mathbf{A},\mathbf{B})]_{jj}.
\]

\begin{theorem}[Sufficient condition for $(\epsilon,\delta)$-DP]
\label{th:dp}

For a colluding worker nodes set $\mathcal T$, define
$s_{\mathcal T}
=
\frac{4r^2t}{\sigma_n^2}
\max_{j\in[K]}
[\mathbf G_{\mathcal T}(\mathbf A, \mathbf B)]_{jj}$.
Then the generalized framework satisfies $(\epsilon,\delta)$-DP against
any coalition of at most $t$ workers whenever
$\delta\geq\max_{|\mathcal T|=t}g(s_{\mathcal T};\epsilon)$, where

\begin{equation}
g(s;\epsilon)
=
\Phi\!\left(
\frac{s-\epsilon}{\sqrt{2s}}
\right)
-
e^\epsilon
\Phi\!\left(
\frac{-s-\epsilon}{\sqrt{2s}}
\right),
\end{equation}

and $\Phi(\cdot)$ denotes the standard Gaussian CDF.
\end{theorem}
\begin{proof}
    See Appendix~ \eqref{proof:theorem2}
\end{proof}

Since $g(s;\epsilon)$ is monotonically increasing in $s$, we immediately
obtain the following worst-case coalition guarantee.

\begin{corollary}
\label{Corollary 2}
The generalized framework satisfies $(\epsilon,\delta)$-DP whenever
$\delta\geq
g\!\left(
\frac{4r^2t}{\sigma_n^2}\Gamma_t(\mathbf A,\mathbf B);
\epsilon
\right)$.
\end{corollary}

We additionally characterize privacy using MIS. Define the worst-case
coalition leakage as
$\eta_C=
\max_{\substack{\mathcal T\subset[N]\\|\mathcal T|=t}}
I(\mathbf X;\mathbf y_{\mathcal T})$.

\begin{proposition}
\label{prop:MIS}
For fixed $N$, $K$, $t$, $r$, and $\sigma_n$, and for
$r=O(\sigma_n)$

\begin{equation}
\label{eq:MIS_second_bound}
\eta_C
\leq
\frac{1}{\ln 2}
\max_{\substack{\mathcal T\subset[N]\\|\mathcal T|=t}}
\mathrm{tr}\!\left(
\mathbf G_{\mathcal T}(\mathbf A,\mathbf B)
\right)
\frac{r^2t}{\sigma_n^2}
+
o\!\left(\frac{r^2}{\sigma_n^2}\right).
\end{equation}

\end{proposition}

\begin{remark}
\label{rem:basis_privacy}
Both privacy metrics depend on the encoding basis through
$\mathbf G_{\mathcal T}(\mathbf A,\mathbf B)$: DP is governed by its worst-case diagonal
entry through $\Gamma_t(\mathbf A,\mathbf B)$, whereas the MIS bound depends on its
trace. In both cases, increasing $\sigma_n^2$ strengthens privacy.
\end{remark}
%\vspace{-0.5cm}
\section{Accuracy-Privacy Aware Basis Design}
\label{sec:basis selection}
We now design the generalized basis functions using the accuracy and
privacy characterizations derived in
Section~\ref{sec:accuracy and privacy}. Recall that the basis functions
are completely determined by the coefficient matrices $\mathbf A$ and $\mathbf B$.
From Remark~\ref{rem:M_accuracy}, reducing
$M(\mathbf A,\mathbf B,\sigma_n)$ improves the basis-dependent component of the
reconstruction-error bound. However, minimizing $M$ alone may yield
$\psi_r(z)=0$, thereby eliminating the noise component and
compromising privacy. On the other hand,
Corollary~\ref{Corollary 2} shows that the DP guarantee depends on the
basis through $\Gamma_t(\mathbf A,\mathbf B)$. We therefore jointly design $\mathbf A$ and $\mathbf B$
to reduce the reconstruction-error bound while satisfying a prescribed privacy requirement. In the following, we proceed with DP as the primary privacy metric since it provides a worst-case guarantee over neighboring datasets, whereas MIS quantifies the average information leakage and depends on the underlying data distribution \cite{DP_VS_MIS}. Nevertheless, the proposed basis-design framework is not tied to DP, and an analogous optimization problem can be formulated using MIS as the privacy metric using Proposition~\ref{prop:MIS}.
Therefore, for a fixed noise level $\sigma_n$, let $s^\star>0$ satisfy
$g(s^\star;\epsilon)=\delta$. The corresponding basis-design problem can be formulated as

\begin{mdframed}
\begin{problem}
\label{optimisation problem}
For the given system parameters $N$, $K$, $t$, $m$, $n$, $r$, $\sigma_n$, $\theta$, $\epsilon$ and $\delta$, solve:

\begin{equation}
\begin{aligned}
(\mathbf A^\star,\mathbf B^\star)
\in\arg\min_{\mathbf A,\mathbf B}\quad
&M(\mathbf A,\mathbf B,\sigma_n)^D\\
\mathrm{s.t.}\quad
&\Gamma_t(\mathbf A,\mathbf B)
\leq
\frac{\sigma_n^2s^\star}{4r^2t}, \quad \sigma_{\min}(\mathbf B)\geq\eta,
\end{aligned}
\label{eq:fixed_noise_problem}
\end{equation}

where $\eta>0$ ensures that the noise-basis coefficient matrix remains
sufficiently separated from singularity.
\end{problem}
\end{mdframed}
The constraint $\sigma_{\min}(\mathbf B)\geq\eta$ also guarantees the nonsingularity of the coalition matrices $\Psi_{\mathcal T}(\mathbf B)$. Rationale behind including this constraint is explained in Lemma~\eqref{lemma:coalition_invertibility}.
\begin{lemma}
\label{lemma:coalition_invertibility}
Suppose that the worker evaluation points
$\{\alpha_i\}_{i=1}^{N}$ are pairwise distinct and
satisfy $q(\alpha_i)\neq 0$ for all $i\in[N]$.
Then, the coalition noise evaluation matrix
$\mathbf{\Psi}_{\mathcal T}(\mathbf B)$ is nonsingular
for every $\mathcal T\subseteq[N]$ with
$|\mathcal T|=t$ if and only if $\mathbf B$ is
nonsingular. Furthermore, imposing
$\sigma_{\min}(\mathbf B)\geq\eta>0$ guarantees $\sigma_{\min}
\bigl(\mathbf{\Psi}_{\mathcal T}(\mathbf B)\bigr)
\geq
\eta\,\sigma_{\min}(\mathbf H_{\mathcal T})
>0,$
where $\mathbf H_{\mathcal T}$ is the
coalition-dependent matrix defined in
Appendix~\ref{app:coalition_matrix_structure}.
\end{lemma}
\begin{proof}
See Appendix~\ref{app:coalition_matrix_structure}.
\end{proof}
Fixing $\sigma_n^2$ a priori, however, unnecessarily restricts the design. For any candidate $(\mathbf A,\mathbf B)$, the minimum noise variance required to satisfy the target DP guarantee can instead be determined directly from $\Gamma_t(\mathbf A,\mathbf B)$.

\begin{theorem}
\label{th:min_noise_var}
Let $s^\star\!>\!\!0$ satisfy $g(s^\star;\epsilon)\!=\!\delta$. For any feasible $\mathbf A$ and $\mathbf B$, the minimum noise variance required for
$(\epsilon,\delta)$-DP is
\begin{equation}
\sigma_{\mathrm{req}}^2(\mathbf A,\mathbf B)
=
\frac{4r^2t}{s^\star}\Gamma_t(\mathbf A,\mathbf B).
\label{eq:required_noise_variance}
\end{equation}
\end{theorem}

Since $M(\mathbf A,\mathbf B,\sigma_n)$ is non-decreasing in $\sigma_n$ for fixed $(\mathbf A,\mathbf B)$, using more noise than required cannot improve the reconstruction-error bound. Hence, it suffices to choose

\begin{equation}
\sigma_n
=
\sigma_{\mathrm{req}}(\mathbf A,\mathbf B)
=
2r
\sqrt{
\frac{t\Gamma_t(\mathbf A,\mathbf B)}{s^\star}
}.
\label{eq:minimum_required_noise}
\end{equation}

Substituting \eqref{eq:minimum_required_noise} into \eqref{eq:M_a}
gives the privacy-calibrated accuracy criterion

\begin{equation}
\label{eq:privacy_calibrated_objective}
\begin{aligned}
\rho(\mathbf A,\mathbf B)
\triangleq \max_{i\in[N]}
\Bigg(
&r\sum_{j=1}^{K}
|\phi_j(\alpha_i;\mathbf A)|\\
&+2r\theta
\sqrt{
\frac{\Gamma_t(\mathbf A,\mathbf B)}{s^\star}
}\\
&\qquad\times
\sum_{r=1}^{t}
|\psi_r(\alpha_i;\mathbf B)|
\Bigg).
\end{aligned}
\end{equation}

\noindent where, by construction,
$\rho(\mathbf A,\mathbf B)=M(\mathbf A,\mathbf B,\sigma_{\mathrm{req}}(\mathbf A, \mathbf B))$.
Thus, minimizing $\rho(\mathbf A,\mathbf B)$ minimizes the basis-dependent component of
the DP compatible reconstruction error bound.
\begin{lemma}
\label{lemma:scale_invariance}
Let $\rho(\mathbf A,\mathbf B)$ denote the
privacy-calibrated objective defined in
\eqref{eq:privacy_calibrated_objective}.
For every admissible pair $(\mathbf A,\mathbf B)$
and every $c\in\mathbb C\setminus\{0\}$, $\rho(\mathbf A,c\mathbf B)
=
\rho(\mathbf A,\mathbf B).$ Therefore, restricting the noise basis coefficient matrix to satisfy $\|\mathbf B\|_F=1$ does not change the attainable objective values of the reduced privacy calibrated basis design problem, provided that only nonsingularity of
$\mathbf B$ is required.
\end{lemma}
\begin{proof}
See Appendix~\ref{app:scale_invariance}.
\end{proof}

By Lemma~\ref{lemma:scale_invariance}, the objective
$\rho(\mathbf A,\mathbf B)$ is invariant to any nonzero scaling of
$\mathbf B$. We therefore impose $\|\mathbf B\|_F=1$ to remove this
redundant degree of freedom. Under this normalization,
$\sigma_{\min}(\mathbf B)\geq\eta$ is feasible if and only if
$0<\eta\leq1/\sqrt{t}$, since
$\|\mathbf B\|_F^2=\sum_{i=1}^{t}\sigma_i^2(\mathbf B)=1$, with equality
attained, for example, by $\mathbf B=\mathbf I_t/\sqrt{t}$.
The resulting privacy calibrated basis design problem is
\begin{mdframed}
\begin{problem}
\label{practical optimisation problem}
For the given system parameters $N$, $K$, $t$, $m$, $n$, $r$, $\theta$, $\epsilon$ and $\delta$, solve:

\begin{equation}
\begin{aligned}
(\mathbf A^\star,\mathbf B^\star)
\in\arg\min_{\mathbf A,\mathbf B}\quad
&\rho(\mathbf A,\mathbf B)\\
\mathrm{s.t.}\quad
&\|\mathbf B\|_F=1, \quad \sigma_{\min}(\mathbf B)\geq\eta .
\end{aligned}
\label{eq:final_opt}
\end{equation}
\end{problem}
\end{mdframed}

Problem~\ref{practical optimisation problem} is non convex \cite{boyd_convex} due to the joint dependence of $\rho(\mathbf A,\mathbf B)$ on the basis parameters and the privacy calibrated noise level. We solve it numerically using a multi-start SQP method \cite{SQP}. The normalization $\|\mathbf B\|_F=1$ removes the scale ambiguity in the noise basis, while $\sigma_{\min}(\mathbf B)\geq\eta$ prevents singular or poorly conditioned noise-basis matrices. Since SQP is a local solver, the resulting solution is not claimed to be globally optimal. For the obtained solution $(\mathbf A^\star,\mathbf B^\star)$, the corresponding DP-compatible noise variance is $(\sigma_n^\star)^2 =\frac{4r^2t}{s^\star} \Gamma_t(\mathbf A^\star,\mathbf B^\star).$ The optimized basis functions then follow from \eqref{phi basis1} and \eqref{psi basis}.

\begin{remark}
\label{rem:surrogate_interpretation}
Under an independent zero mean relative roundoff model, the expected squared reconstruction error scales with $\sum_{i\in\mathcal S}|w_{j,i}|^2\|f(\mathbf u(\alpha_i))\|_F^2$, where $\mathcal S$ is the reconstruction worker set and $w_{j,i}$ the corresponding interpolation weights. Thus, reducing worker-computation magnitudes reduces the average numerical error, motivating $\rho(\mathbf A,\mathbf B)$ as a function-agnostic surrogate that controls the encoded share magnitude.
\end{remark}

\section{Experimental Results}
\label{sec:exp results}

We evaluate the proposed framework for
$N=15$, $K=5$, $t=3$, $m=2000$, $n=100$, $r=1000$, and $D=2$,
with $f(\mathbf X)=\mathbf X^T\mathbf X$ and i.i.d.\
$\mathcal N(0,1)$ entries clipped to $[-r,r]$.
The generalized construction uses $K$ equally spaced data-recovery
points on the circle of radius $\beta$, while conventional ALCC retains
its original $K+t$ equally spaced interpolation points on the same
circle. We set $\theta=6$, giving
$(1-e^{-36})^{6\times10^5}\geq1-1.39\times10^{-10}$, and
$\eta=0.1<1/\sqrt3$.
Problem~\ref{practical optimisation problem} is solved by multi-start
SQP with four initializations, initialized at
$a_{j,n}=0$ and $b_{r,n}=\delta_{rn}/\sqrt t$.
We use $\delta=10^{-4}$ and
$\epsilon\in\{0.25,0.5,1\}$ for DP, and
$\eta_C\in\{10^{-2},10^{-4},10^{-5}\}$ for MIS.
The normalized Frobenius error is averaged over $10^3$ trials.

\begin{figure}[t]
    \centering
    \begin{subfigure}{0.49\columnwidth}
        \centering
        \includegraphics[
            width=\linewidth,
            trim={5mm 4mm 5mm 4mm},
            clip
        ]{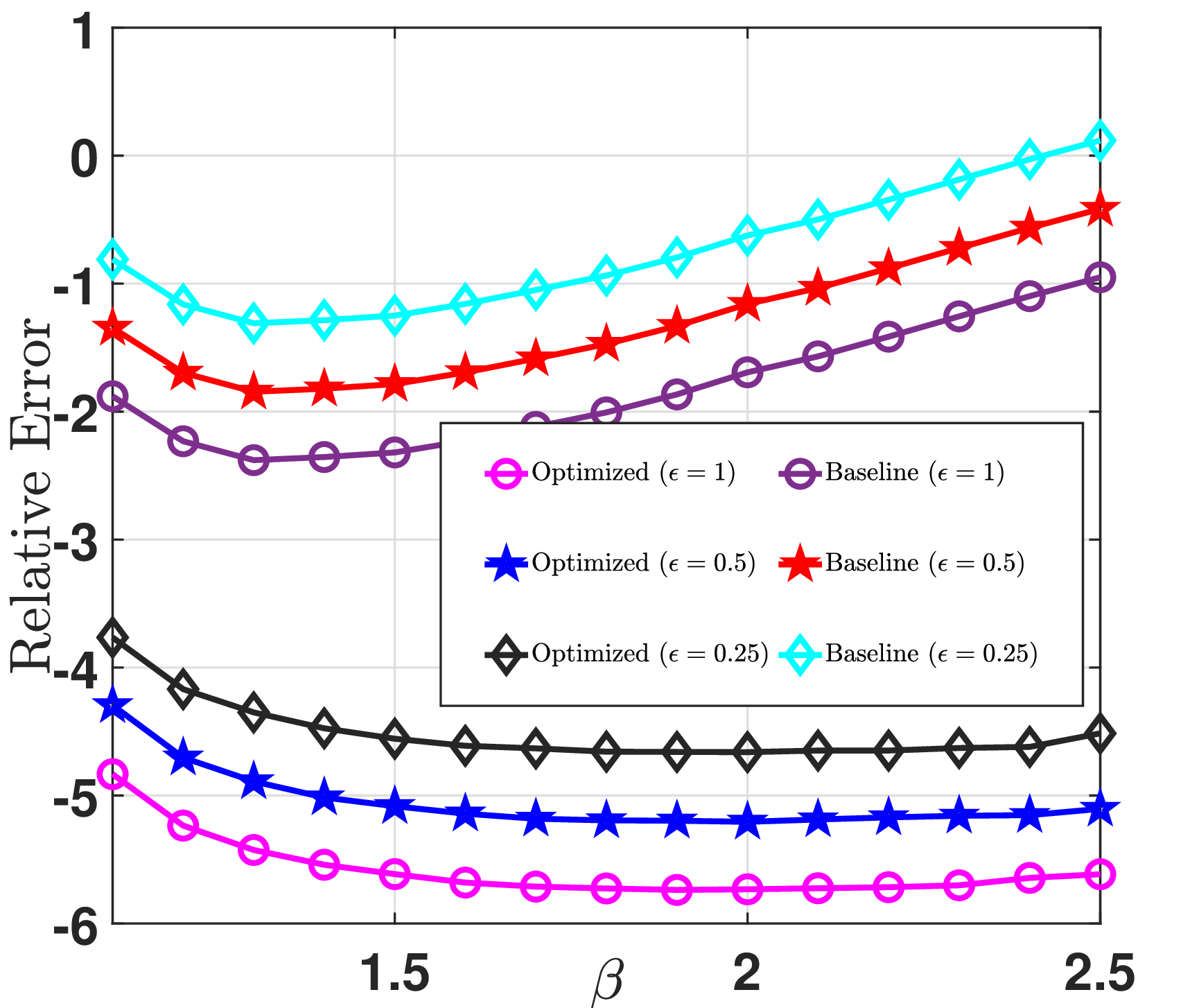}
        \caption{Differential privacy.}
        \label{fig:DPAccuracy}
    \end{subfigure}
    \hfill
    \begin{subfigure}{0.49\columnwidth}
        \centering
        \includegraphics[
            width=\linewidth,
            trim={5mm 4mm 5mm 4mm},
            clip
        ]{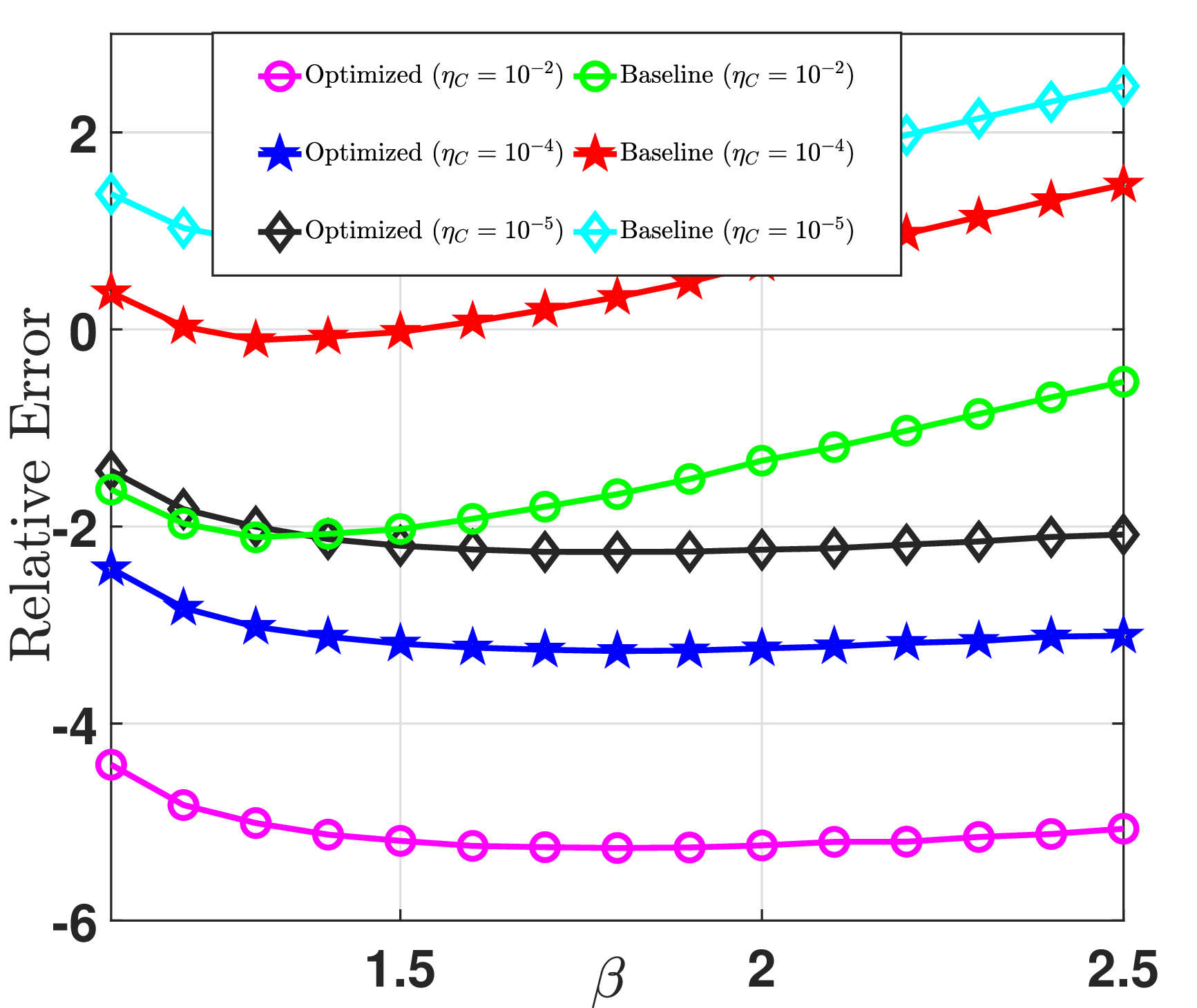}
        \caption{Mutual information security.}
        \label{fig:MISAccuracy}
    \end{subfigure}
    \vspace{-2mm}
    \caption{Reconstruction error of the optimized basis and
    conventional ALCC under DP and MIS.}
    \label{fig:AccuracyComparison}
    \vspace{-1mm}
\end{figure}
Fig.~\ref{fig:AccuracyComparison} shows lower reconstruction error for
the proposed design over the considered range, with the best performance around $\beta\in[1.7,2]$. Including MIS alongside DP further indicates that the basis-design principle is not tied to a particular privacy metric. We next consider the cubic computation $f(\mathbf X)=\mathbf X(\mathbf X^T\mathbf X)$ with $D=3$, $N=7$, $K=1$, $t=2$, $m=10^4$, $n=100$, $r=1000$, $\theta=6$, $\delta=10^{-4}$, and $\epsilon\in\{0.5,1.5\}$, to examine whether the gain persists for a higher-degree computation.
\begin{figure}[t]
    \centering
    \begin{subfigure}{0.49\columnwidth}
        \centering
        \includegraphics[width=\linewidth]{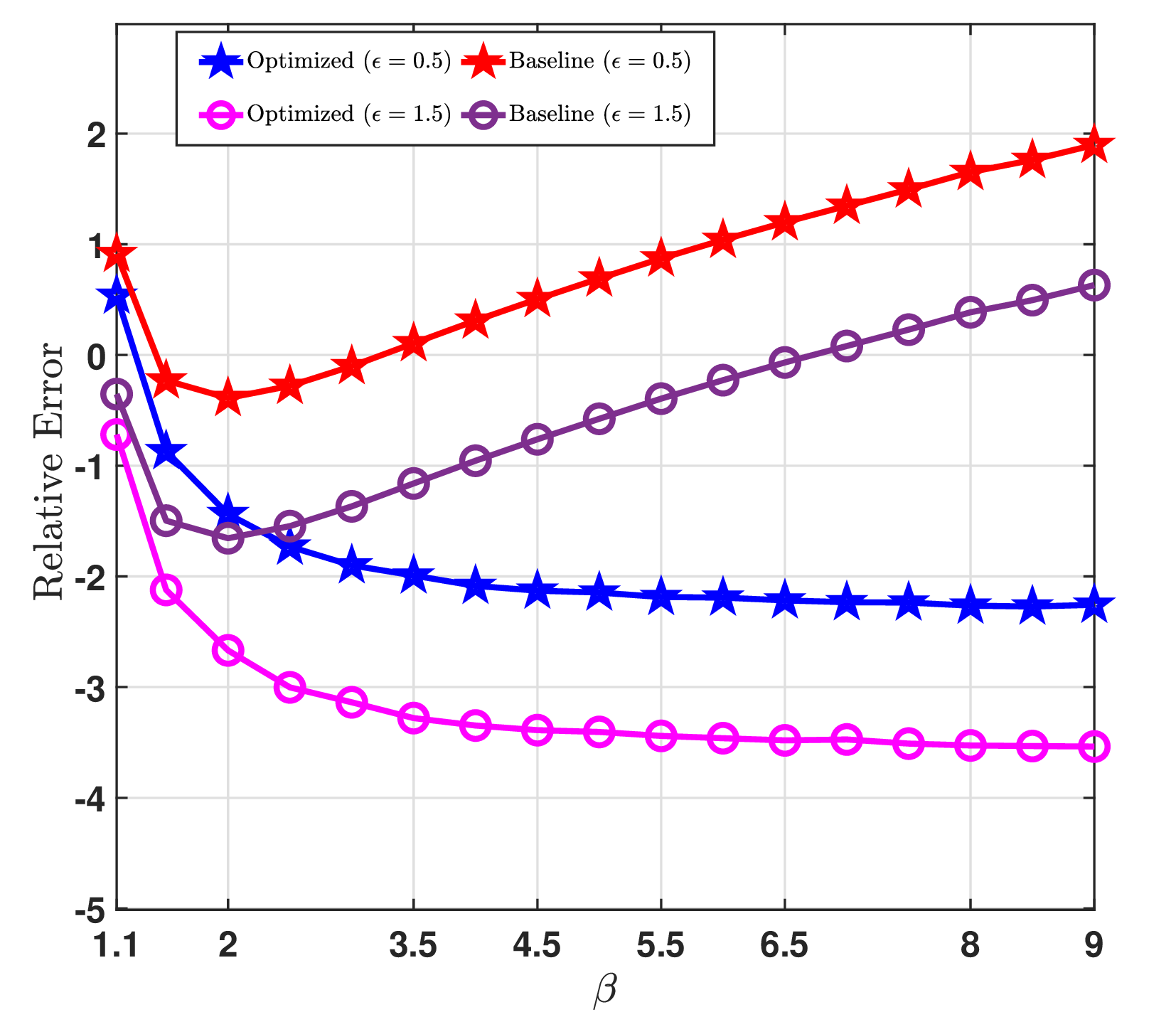}
        \caption{Relative reconstruction error.}
        \label{fig:DPAccuracy2}
    \end{subfigure}
    \hfill
    \begin{subfigure}{0.49\columnwidth}
        \centering
        \includegraphics[width=\linewidth]{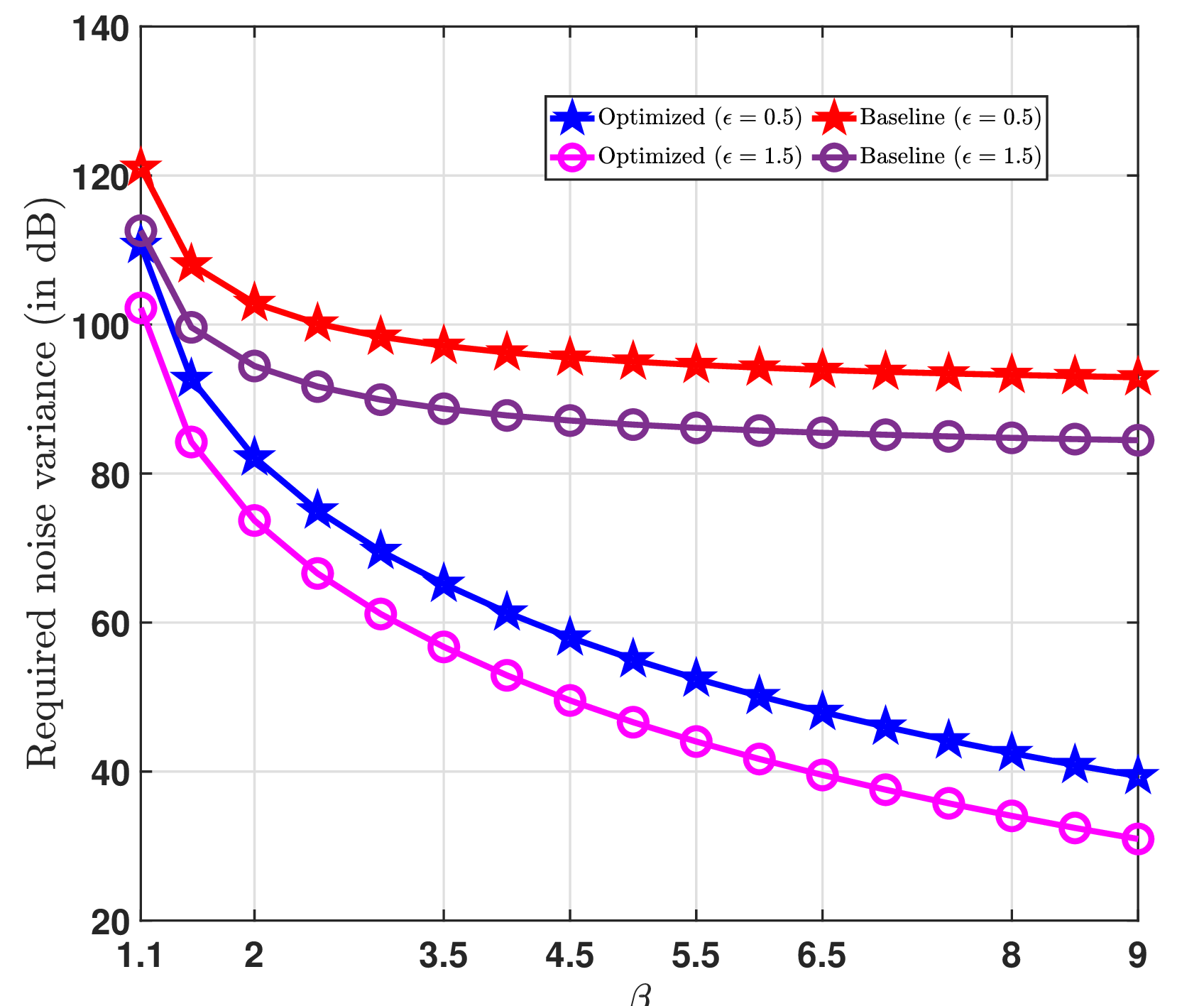}
        \caption{Required noise variance in dB.}
        \label{fig:NoiseVarRe2}
    \end{subfigure}
    \vspace{-1mm}
    \caption{Accuracy and required noise for
    $f(\mathbf X)=\mathbf X(\mathbf X^T\mathbf X)$ under
    $(\epsilon,\delta)$-DP.}
    \label{fig:AccuracyComparisonExpt2}
  \vspace{-0.5cm}
\end{figure}
Fig.~\ref{fig:AccuracyComparisonExpt2} shows that the gain over conventional ALCC persists for the cubic computation. The optimized basis also requires substantially less DP-compatible noise over the considered $\beta$ range, indicating a more favorable privacy geometry. Together, Figs.~\ref{fig:AccuracyComparison} and \ref{fig:AccuracyComparisonExpt2} provide empirical evidence that the proposed design is not restricted to a particular privacy metric or polynomial degree.

\section{Discussion}
\label{sec:discussion}
This work treats the encoding basis in analog coded computing as a joint
accuracy--privacy design variable while preserving exact recovery, with numerical gains over conventional ALCC under both DP and MIS constraints. The present optimization minimizes a basis-dependent surrogate and is solved using nonconvex multi-start SQP, future work may consider tighter end-to-end error formulations and \emph{function-dependent} objectives that exploit the structure of the target computation $f$. Further directions include joint basis and evaluation-point design, e.g., using Chebyshev nodes, and extending the proposed parameterization to approximation-based schemes for non-polynomial functions, including Berrut interpolants \cite{b6,b7} and B-spline based reconstruction \cite{bspline}. It is an interesting direction to even see how this framework works under the assumption of profiled worker nodes as considered in \cite{ncc}.

%%%%%%%%%%%%%%%%
\appendices

\section{Proofs of the Accuracy and Privacy Results}
\label{app:accuracy_privacy_proofs}

\subsection{Proof of Theorem~\ref{th:accuracy}}
\begin{proof}
\label{proof:theorem1}
Consider the encoding polynomial in \eqref{eq:New encoding} together with the parameterized basis functions defined in \eqref{phi basis1} and \eqref{psi basis}. By the interpolation and cancellation constraints, $\phi_j(\beta_\ell)=\delta_{j\ell}$ and $\psi_r(\beta_\ell)=0$, for all $j,\ell\in[K]$ and $r\in[t]$. Hence, for every $j\in[K]$,
\begin{align}
\mathbf u(\beta_j)
&=
\sum_{\ell=1}^{K}\phi_\ell(\beta_j)\mathbf X_\ell
+
\sum_{r=1}^{t}\psi_r(\beta_j)\mathbf N_r
=
\mathbf X_j .
\label{eq:recover_Xj}
\end{align}
Therefore, if $g(z)=f(\mathbf u(z))$, then \eqref{eq:recover_Xj} implies
\[
g(\beta_j)=f(\mathbf X_j), \qquad j\in[K].
\]
Since $\phi_j,\psi_r\in\mathcal P_{K+t-1}$, we have
$\deg \mathbf u(z)\le K+t-1$. As $f$ is a polynomial of degree at most $D$, it follows that
$\deg g(z)\le D(K+t-1)=\widetilde D$. Hence, $g(z)$ can be reconstructed from any $\widetilde D+1$ worker evaluations.

We next upper bound the magnitude of the encoded shares at the worker nodes. The noise matrices are not assumed to be deterministically bounded. Instead, consider the event $\mathcal E_\theta$ defined as:

\[
\label{eq:noise_event}
\mathcal E_\theta
\triangleq
\left\{
\max_{r\in[t]}
\|\mathbf N_r\|_{\max}
\leq
\frac{\theta\sigma_n}{\sqrt t}
\right\}.
\]

Since each noise entry satisfies
$(\mathbf N_r)_{pq}\sim\mathcal{CN}(0,\sigma_n^2/t)$, we have $\mathbb{P}\left(
|(\mathbf N_r)_{pq}|>
\frac{\theta\sigma_n}{\sqrt t}
\right)
=
e^{-\theta^2}.$ Since the $tmn$ noise entries are independent, the event
$\mathcal E_\theta$ occurs with probability
\[
\mathbb{P}(\mathcal E_\theta)
=
\left(1-e^{-\theta^2}\right)^{tmn}
\geq
1-tmn e^{-\theta^2},
\]
where the last inequality follows from the union bound. Whenever
$\mathcal E_\theta$ occurs,
$\|\mathbf N_r\|_{\max}\leq\theta\sigma_n/\sqrt t$ for every
$r\in[t]$. Therefore, using $\|\mathbf X_j\|_{\max}\le r$, for every worker evaluation point $\alpha_i$,
\[
\|\mathbf u(\alpha_i)\|_{\max}
\le
r\sum_{j=1}^{K}|\phi_j(\alpha_i)|
+
\frac{\theta\sigma_n}{\sqrt t}
\sum_{r=1}^{t}|\psi_r(\alpha_i)|
\le M,
\]
where $M$ is defined in \eqref{eq:M_a}.
Going ahead, fix an output entry $(p,q)$. Since $f_{pq}$ is a polynomial of degree at most $D$ in the $mn$ entries of $\mathbf u(\alpha_i)$, and since the magnitude of each coefficient is bounded by $c$, on $\mathcal E_\theta$ we have
\[
|g_{pq}(\alpha_i)|
\le
c\sum_{d=0}^{D}
{mn+d-1\choose d}M^d .
\]
The quantity ${mn+d-1\choose d}$ counts the number of monomials of total degree $d$ in $mn$ variables. Using the standard bound
${mn+d-1\choose d}\le (emn)^d$, we obtain
\[
|g_{pq}(\alpha_i)|
\le
c\sum_{d=0}^{D}(emn)^dM^d
\le
c(D+1)\max\{1,(emn)^DM^D\}.
\]
Under the assumption $(emn)M\ge1$, this reduces to
\begin{equation}
\label{eq:g_bound}
|g_{pq}(\alpha_i)|
\le
c(D+1)(emn)^D M^D .
\end{equation}
Let $\mathbf g_{pq}\in\mathbb C^{\widetilde D+1}$ denote the vector of returned evaluations corresponding to entry $(p,q)$ from the workers used for reconstruction. Using \eqref{eq:g_bound},
\begin{equation}
\label{eq:g_norm}
\|\mathbf g_{pq}\|_2
\le
\sqrt{\widetilde D+1}\,
c(D+1)(emn)^D M^D .
\end{equation}
We now characterize the interpolation error introduced by finite-precision arithmetic. Let
$g_{pq}(z)=\sum_{\ell=0}^{\widetilde D}a_{\ell,pq}z^\ell$ and
$\mathbf a_{pq}=(a_{0,pq},a_{1,pq},\ldots,a_{\widetilde D,pq})^T$.
The interpolation step can be written as $\mathbf g_{pq}
=
\mathcal V\mathbf a_{pq},$ where $\mathcal V$ is the Vandermonde matrix formed from the selected $\widetilde D+1$ worker evaluation points. Since $\mathbf a_{pq}=\mathcal V^{-1}\mathbf g_{pq}$,
\begin{equation}
\label{eq:a_bound}
\|\mathbf a_{pq}\|_2
\le
\|\mathcal V^{-1}\|_2\|\mathbf g_{pq}\|_2
=
\frac{1}{\sigma_{\min}(\mathcal V)}
\|\mathbf g_{pq}\|_2 .
\end{equation}
Further, from the finite-precision interpolation model in
\eqref{eq:finite_precision_interpolation},
\[
\|\widehat{\mathbf a}_{pq}-\mathbf a_{pq}\|_2
\le
\kappa_2(\mathcal V)\|\mathbf a_{pq}\|_2 2^{-b_m}.
\]
Combining this with \eqref{eq:a_bound} and \eqref{eq:g_norm}, we obtain
\begin{equation}
\label{eq:coeff_final}
\|\widehat{\mathbf a}_{pq}-\mathbf a_{pq}\|_2
\le
\frac{\kappa_2(\mathcal V)}{\sigma_{\min}(\mathcal V)}
\sqrt{\widetilde D+1}\,
c(D+1)(emn)^D M^D 2^{-b_m}.
\end{equation}
Finally, we propagate the coefficient perturbation to the reconstructed function value at $\beta_j$. Since the reconstructed value is obtained by evaluating the interpolated polynomial at $\beta_j$,
\[
\widehat f(\mathbf X_j)_{pq}-f(\mathbf X_j)_{pq}
=
\sum_{\ell=0}^{\widetilde D}
(\widehat a_{\ell,pq}-a_{\ell,pq})\beta_j^\ell .
\]
Applying the Cauchy-Schwarz inequality gives
\begin{equation}
\label{eq:cs_bound}
|\widehat f(\mathbf X_j)_{pq}-f(\mathbf X_j)_{pq}|
\le
\|(1,\beta_j,\ldots,\beta_j^{\widetilde D})\|_2
\|\widehat{\mathbf a}_{pq}-\mathbf a_{pq}\|_2 .
\end{equation}
Since $|\beta_j|=\beta$,
\[
\|(1,\beta_j,\ldots,\beta_j^{\widetilde D})\|_2^2
=
\sum_{\ell=0}^{\widetilde D}\beta^{2\ell}
=
\Theta_\beta .
\]
Substituting this and \eqref{eq:coeff_final} into \eqref{eq:cs_bound}, we obtain
\begin{align}
|\widehat f(\mathbf X_j)_{pq}-f(\mathbf X_j)_{pq}|
\le\;&
\frac{
\sqrt{\Theta_\beta}\kappa_2(\mathcal V)
\sqrt{\widetilde D+1}
}{
\sigma_{\min}(\mathcal V)
}
\notag\\
&\times
c(D+1)(emn)^D M^D 2^{-b_m}.
\label{eq:absolute_error_final}
\end{align}

Thus, \eqref{eq:absolute_error_final} holds whenever $\mathcal E_\theta$ occurs. Since,
\[
\mathbb{P}(\mathcal E_\theta)
=
\left(1-e^{-\theta^2}\right)^{tmn}
\ge
1-tmn e^{-\theta^2}.
\]
Therefore, the bound in \eqref{eq:absolute_error_final} holds with the probability stated in Theorem~\ref{th:accuracy}, which completes the proof.
\end{proof}
%\noindent\rule{\linewidth}{0.6pt}
\subsection{Proof of Lemma~\ref{lem:choice_theta}}
\label{proof:lemma1}
\begin{proof}
From the proof of Theorem~\ref{th:accuracy}, the reconstruction
error bound in \eqref{accuracy bound} holds whenever the event
\[
\mathcal E_\theta = \left\{\max_{r\in[t]}\|\mathbf N_r\|_{\max}\leq\frac{\theta\sigma_n}{\sqrt t}\right\}
\]
occurs. Since the $tmn$ entries of the noise matrices are independent and each entry satisfies
$(\mathbf N_r)_{pq}\sim\mathcal{CN}(0,\sigma_n^2/t)$, we have
\[
\Pr(\mathcal E_\theta)
=
\left(1-e^{-\theta^2}\right)^{tmn}.
\]
Therefore, in order for the reconstruction error bound in
\eqref{accuracy bound} to hold with probability at least $1-\zeta$, it is sufficient to require
\[
\left(1-e^{-\theta^2}\right)^{tmn}
\geq
1-\zeta.
\]
Since both sides lie in $[0,1]$, this is equivalent to
\[
1-e^{-\theta^2}
\geq
(1-\zeta)^{1/(tmn)},
\]
or,
\[
e^{-\theta^2}
\leq
1-(1-\zeta)^{1/(tmn)}.
\]
Taking the logarithm on both sides and using $\theta>0$ gives
\[
\theta
\geq
\sqrt{
-\log\left(
1-(1-\zeta)^{1/(tmn)}
\right)
}.
\]
Hence, the smallest value of $\theta$ satisfying the above condition is
\[
\theta_{\min}
=
\sqrt{
-\log\left(
1-(1-\zeta)^{1/(tmn)}
\right)
},
\]
and any $\theta\geq\theta_{\min}$ ensures that the reconstruction error bound in \eqref{accuracy bound} holds with probability at least $1-\zeta$.
\end{proof}

%\noindent\rule{\linewidth}{0.6pt}
\subsection{Proof of Theorem~\ref{th:dp}}
\label{proof:theorem2}
\begin{proof}
Fix a colluding set $\mathcal T\subset[N]$ with
$|\mathcal T|=t$. From the observation model in
\eqref{coalition model}, the encoded shares available to the colluding
workers are
\[
\mathbf y_{\mathcal T}
=
(\mathbf \Phi_{\mathcal T}\otimes \mathbf I_{mn})\mathbf x
+
(\mathbf \Psi_{\mathcal T}\otimes \mathbf I_{mn})\mathbf n.
\]
Define
$\mathbf A_{\mathcal T}
=
\mathbf\Phi_{\mathcal T}\otimes\mathbf I_{mn}$
and
$\boldsymbol\eta_{\mathcal T}
=
(\mathbf\Psi_{\mathcal T}\otimes\mathbf I_{mn})\mathbf n$.
Then
\begin{equation}
\label{eq:dp_linear_model}
\mathbf y_{\mathcal T}
=
\mathbf A_{\mathcal T}\mathbf x
+
\boldsymbol\eta_{\mathcal T}.
\end{equation}

Since the entries of the noise matrices are independently drawn from a
circularly symmetric complex Gaussian distribution with variance
$\sigma_n^2/t$, we have $\boldsymbol\eta_{\mathcal T}
\sim
\mathcal{CN}
\left(
0,\mathbf\Sigma_{\mathcal T}
\right),$
where $\mathbf\Sigma_{\mathcal T}
=
\frac{\sigma_n^2}{t}
(\mathbf\Psi_{\mathcal T}\mathbf\Psi_{\mathcal T}^{H})
\otimes
\mathbf I_{mn}.$ Now, consider two neighboring datasets
$\mathbf X\sim\mathbf X'$ under the definition that they differ in at most one scalar entry, and let $\mathbf x,\mathbf x'$ denote their vectorized forms. Conditioned on $\mathbf x$, the observation in \eqref{eq:dp_linear_model} follows $P =\mathcal{CN}\left(\mathbf A_{\mathcal T}\mathbf x,\mathbf\Sigma_{\mathcal T} \right),$ whereas conditioned on $\mathbf x'$, it follows $Q =\mathcal{CN}\left(\mathbf A_{\mathcal T}\mathbf x',\mathbf\Sigma_{\mathcal T}\right).$ For compactness, define $\boldsymbol\mu=\mathbf A_{\mathcal T}\mathbf x,\\boldsymbol\mu'=\mathbf A_{\mathcal T}\mathbf x',\\mathbf d=\boldsymbol\mu-\boldsymbol\mu'.$ The privacy-loss random variable is
\[
L(\mathbf y)=\log\frac{p(\mathbf y)}{q(\mathbf y)}.
\]
Substituting the complex Gaussian densities of $P$ and $Q$ and simplifying gives
\begin{align}
L(\mathbf y)
&=
-(\mathbf y-\boldsymbol\mu)^H
\mathbf\Sigma_{\mathcal T}^{-1}
(\mathbf y-\boldsymbol\mu)
+
(\mathbf y-\boldsymbol\mu')^H
\mathbf\Sigma_{\mathcal T}^{-1}
(\mathbf y-\boldsymbol\mu')
\notag\\
&=
2\Re\!\left\{
\mathbf d^H
\mathbf\Sigma_{\mathcal T}^{-1}
(\mathbf y-\boldsymbol\mu')
\right\}
-
\mathbf d^H
\mathbf\Sigma_{\mathcal T}^{-1}
\mathbf d .
\label{eq:dp_privacy_loss}
\end{align}

Further, define
\begin{equation}
\label{eq:dp_s_pair}
s(\mathbf x,\mathbf x')
=
\mathbf d^H
\mathbf\Sigma_{\mathcal T}^{-1}
\mathbf d
=
(\mathbf x-\mathbf x')^H
\mathbf A_{\mathcal T}^{H}
\mathbf\Sigma_{\mathcal T}^{-1}
\mathbf A_{\mathcal T}
(\mathbf x-\mathbf x').
\end{equation}
We next characterize the distribution of $L(\mathbf y)$. If
$\mathbf y\sim Q$, then
$\mathbf y-\boldsymbol\mu'
\sim\mathcal{CN}(0,\mathbf\Sigma_{\mathcal T})$, and therefore
\[
\mathbf d^H
\mathbf\Sigma_{\mathcal T}^{-1}
(\mathbf y-\boldsymbol\mu')
\sim
\mathcal{CN}
\left(
0,s(\mathbf x,\mathbf x')
\right).
\]
Since the real part of a circularly symmetric complex Gaussian random
variable $\mathcal{CN}(0,s)$ has variance $s/2$, it follows from
\eqref{eq:dp_privacy_loss} that
\[
L\mid Q
\sim
\mathcal N
\left(
-s(\mathbf x,\mathbf x'),
2s(\mathbf x,\mathbf x')
\right).
\]
Similarly, if $\mathbf y\sim P$, then
$\mathbf y-\boldsymbol\mu'
=
(\mathbf y-\boldsymbol\mu)+\mathbf d$, and hence
\[
L\mid P
\sim
\mathcal N
\left(
s(\mathbf x,\mathbf x'),
2s(\mathbf x,\mathbf x')
\right).
\]
For a fixed neighboring pair $(\mathbf x,\mathbf x')$, the smallest
$\delta$ satisfying $P(\mathcal S)
\leq
e^\epsilon Q(\mathcal S)+\delta$
for every measurable set $\mathcal S$ is
\[
\delta(\epsilon)
=
\sup_{\mathcal S}
\left[
P(\mathcal S)-e^\epsilon Q(\mathcal S)
\right].
\]
We now show that a maximizing set is $\mathcal S^\star
=
\{\mathbf y:L(\mathbf y)>\epsilon\}.$
For any measurable set $\mathcal S$,
\begin{align*}
P(\mathcal S)-e^\epsilon Q(\mathcal S)
&=
\int_{\mathcal S}
\left(
p(\mathbf y)-e^\epsilon q(\mathbf y)
\right)d\mathbf y\\
&=
\int_{\mathcal S}
\left(
e^{L(\mathbf y)}-e^\epsilon
\right)
q(\mathbf y)\,d\mathbf y.
\end{align*}
Define $h(\mathbf y)
=
\left(
e^{L(\mathbf y)}-e^\epsilon
\right)q(\mathbf y).$
Since the Gaussian density satisfies $q(\mathbf y)>0$ for every
$\mathbf y$, $h(\mathbf y)>0
\ \Longleftrightarrow \
L(\mathbf y)>\epsilon.$
Hence,
\[
\mathcal S^\star
=
\{\mathbf y:h(\mathbf y)>0\}
=
\{\mathbf y:L(\mathbf y)>\epsilon\}.
\]
For an arbitrary measurable set $\mathcal S$, $\mathcal S^\star
=
(\mathcal S^\star\cap\mathcal S)
\cup
(\mathcal S^\star\setminus\mathcal S),$
and $\mathcal S
=
(\mathcal S^\star\cap\mathcal S)
\cup
(\mathcal S\setminus\mathcal S^\star),$
where both unions are disjoint. Therefore,
\[
\int_{\mathcal S^\star}h(\mathbf y)d\mathbf y
-
\int_{\mathcal S}h(\mathbf y)d\mathbf y
=
\int_{\mathcal S^\star\setminus\mathcal S}
h(\mathbf y)d\mathbf y
-
\int_{\mathcal S\setminus\mathcal S^\star}
h(\mathbf y)d\mathbf y.
\]
By definition of $\mathcal S^\star$,
$h(\mathbf y)>0$ on
$\mathcal S^\star\setminus\mathcal S$, whereas
$h(\mathbf y)\leq0$ on
$\mathcal S\setminus\mathcal S^\star$. Hence,
\[
\int_{\mathcal S^\star}h(\mathbf y)d\mathbf y
\geq
\int_{\mathcal S}h(\mathbf y)d\mathbf y.
\]
Equivalently, $P(\mathcal S)-e^\epsilon Q(\mathcal S)
\leq
P(\mathcal S^\star)-e^\epsilon Q(\mathcal S^\star),$
which proves that $\mathcal S^\star$ is a maximizing set. Consequently, $\delta(\epsilon)
=
P(L>\epsilon)-e^\epsilon Q(L>\epsilon).$
Using $L\mid P\sim\mathcal N(s,2s),
\qquad
L\mid Q\sim\mathcal N(-s,2s),$
where $s=s(\mathbf x,\mathbf x')$, we obtain
\[
P(L>\epsilon)
=
\Phi\!\left(
\frac{s-\epsilon}{\sqrt{2s}}
\right),
\]
\[
Q(L>\epsilon)
=
\Phi\!\left(
\frac{-s-\epsilon}{\sqrt{2s}}
\right).
\]
Hence, for the fixed neighboring pair,
\[
\delta(\epsilon)
=
\Phi\!\left(
\frac{s-\epsilon}{\sqrt{2s}}
\right)
-
e^\epsilon
\Phi\!\left(
\frac{-s-\epsilon}{\sqrt{2s}}
\right)
=
g(s;\epsilon).
\]
Taking the worst case over all neighboring pairs for the fixed
coalition $\mathcal T$, define
\[
s_{\mathcal T}
=
\sup_{\mathbf X\sim\mathbf X'}
(\mathbf x-\mathbf x')^H
\mathbf A_{\mathcal T}^{H}
\mathbf\Sigma_{\mathcal T}^{-1}
\mathbf A_{\mathcal T}
(\mathbf x-\mathbf x').
\]
Since $g(s;\epsilon)$ is increasing in $s$, the coalition-level
privacy profile is $\delta_{\mathcal T}(\epsilon)
=
g(s_{\mathcal T};\epsilon).$
It remains to compute $s_{\mathcal T}$. Using the Kronecker structure,
\[
\mathbf A_{\mathcal T}
=
\mathbf\Phi_{\mathcal T}\otimes\mathbf I_{mn}
\]
\[
\mathbf\Sigma_{\mathcal T}
=
\frac{\sigma_n^2}{t}
(\mathbf\Psi_{\mathcal T}\mathbf\Psi_{\mathcal T}^{H})
\otimes
\mathbf I_{mn},
\]
we obtain
\begin{equation}
\label{eq:dp_quad_simplified}
\mathbf A_{\mathcal T}^{H}
\mathbf\Sigma_{\mathcal T}^{-1}
\mathbf A_{\mathcal T}
=
\frac{t}{\sigma_n^2}
\left(
\mathbf\Phi_{\mathcal T}^{H}
(\mathbf\Psi_{\mathcal T}\mathbf\Psi_{\mathcal T}^{H})^{-1}
\mathbf\Phi_{\mathcal T}
\right)
\otimes
\mathbf I_{mn}.
\end{equation}
Defining $\mathbf G_{\mathcal T}
=
\mathbf\Phi_{\mathcal T}^{H}
(\mathbf\Psi_{\mathcal T}\mathbf\Psi_{\mathcal T}^{H})^{-1}
\mathbf\Phi_{\mathcal T},$
equation \eqref{eq:dp_quad_simplified} becomes
\begin{equation}
\label{eq:dp_quad_GT}
\mathbf A_{\mathcal T}^{H}
\mathbf\Sigma_{\mathcal T}^{-1}
\mathbf A_{\mathcal T}
=
\frac{t}{\sigma_n^2}
(\mathbf G_{\mathcal T}\otimes\mathbf I_{mn}).
\end{equation}
Using the definition that two neighboring datasets differ in at most one scalar entry. Hence, there exist $j\in[K]$, $(p,q)\in[m]\times[n]$, and a scalar $\mathbf A$ with $|a|\leq2r$, such that $\mathbf x-\mathbf x'
=
\mathbf e_j\otimes a\mathbf e_{pq},$
where $\mathbf e_j$ is the $j$-th canonical basis vector and
$\mathbf e_{pq}$ is the canonical basis vector corresponding to entry
$(p,q)$ after vectorization. Substituting
\eqref{eq:dp_quad_GT} into \eqref{eq:dp_s_pair}, we obtain
\begin{equation}
\label{eq:privacy_quadratic_form}
\begin{aligned}
s(\mathbf x,\mathbf x')
&=
\frac{t}{\sigma_n^2}
(\mathbf e_j\otimes a\mathbf e_{pq})^H
(\mathbf G_{\mathcal T}\otimes\mathbf I_{mn})
(\mathbf e_j\otimes a\mathbf e_{pq})\\
&=
\frac{t}{\sigma_n^2}
(\mathbf G_{\mathcal T})_{jj}|a|^2.
\end{aligned}
\end{equation}
Since $|a|\leq2r$, taking the supremum over all neighboring datasets
gives
\[
\label{eq:dp_sT_final}
s_{\mathcal T}
=
\frac{4r^2t}{\sigma_n^2}
\max_{j\in[K]}
(\mathbf G_{\mathcal T})_{jj}.
\]
Therefore, for every coalition $\mathcal T$ of size $t$, the
corresponding privacy parameter is $\delta_{\mathcal T}(\epsilon)
=
g(s_{\mathcal T};\epsilon).$
Thus, the proposed framework satisfies $(\epsilon,\delta)$-DP
uniformly over all coalitions of size $t$ whenever
\[
\delta
\geq
\max_{\substack{\mathcal T\subset[N]\\|\mathcal T|=t}}
g(s_{\mathcal T};\epsilon).
\]
Finally, consider any coalition $\mathcal S\subset[N]$ with
$|\mathcal S|<t$. Choose a set $\mathcal T\supset\mathcal S$ with
$|\mathcal T|=t$. The observations available to the workers in
$\mathcal S$ are obtained from $\mathbf y_{\mathcal T}$ by discarding
the components corresponding to the workers in $\mathcal T\setminus\mathcal S$. Hence, the observation of $\mathcal S$ is a deterministic post-processing of the observation of $\mathcal T$. By the post-processing property of differential privacy, the same $(\epsilon,\delta)$ guarantee therefore holds for every coalition of size less than $t$. This proves the result for any set of at most $t$ colluding worker nodes.
\end{proof}
%\noindent\rule{\linewidth}{0.6pt}
\subsection{Proof of Theorem~\ref{th:min_noise_var}}
\begin{proof}
From Corollary~\ref{Corollary 2}, the proposed framework satisfies \((\epsilon,\delta)\)-DP whenever
\[
\delta
\ge
g\!\left(
\frac{4r^2t}{\sigma_n^2}\Gamma_t;\epsilon
\right),
\]
where $g(s;\epsilon)
=
\Phi\!\left(
\frac{s-\epsilon}{\sqrt{2s}}
\right)
-
e^\epsilon
\Phi\!\left(
\frac{-s-\epsilon}{\sqrt{2s}}
\right).$ 

Define
\begin{equation}
\label{eq:s_prot_def}
s_{\mathrm{prot}}
=
\frac{4r^2t}{\sigma_n^2}\Gamma_t .
\end{equation}
Then the above DP condition can be written equivalently as
\begin{equation}
\label{eq:delta_condition}
g(s_{\mathrm{prot}};\epsilon)\le \delta .
\end{equation}
For a fixed value of \(\epsilon>0\), the Gaussian privacy profile
\(g(s;\epsilon)\) is monotonically increasing in \(s\). Let
\(s^\star>0\) denote the unique solution satisfying $g(s^\star;\epsilon)=\delta .$
By monotonicity of \(g(s;\epsilon)\) (which is proved subsequently), the condition
in \eqref{eq:delta_condition} holds if and only if
\begin{equation}
\label{eq:s_prot_bound}
s_{\mathrm{prot}}\le s^\star .
\end{equation}
Substituting \eqref{eq:s_prot_def} into
\eqref{eq:s_prot_bound}, we obtain
\[
\frac{4r^2t}{\sigma_n^2}\Gamma_t
\le
s^\star .
\]
Rearranging the above inequality yields
\begin{equation}
\label{eq:sigma_bound_final}
\sigma_n^2
\ge
\frac{4r^2t}{s^\star}\Gamma_t .
\end{equation}
Therefore, whenever \eqref{eq:sigma_bound_final} holds, the proposed
framework satisfies \((\epsilon,\delta)\)-DP.
\end{proof}

\subsection{Proof that \(g(s;\epsilon)\) is monotonically increasing in \(s>0\) for fixed \(\epsilon>0\)}
\begin{proof}
Recall that $g(s;\epsilon)
=
\Phi\!\left(
\frac{s-\epsilon}{\sqrt{2s}}
\right)
-
e^\epsilon
\Phi\!\left(
\frac{-s-\epsilon}{\sqrt{2s}}
\right),\ s>0.$
Define
\[
a(s)=\frac{s-\epsilon}{\sqrt{2s}},
\qquad
b(s)=\frac{-s-\epsilon}{\sqrt{2s}}.
\]
Then $g(s;\epsilon)=\Phi(a(s))-e^\epsilon\Phi(b(s)).$
Differentiating with respect to \(s\), and using
\(\Phi'(x)=\varphi(x)\), where \(\varphi(x)\) is the standard Gaussian
density, gives
\[
\frac{d}{ds}g(s;\epsilon)
=
\varphi(a(s))a'(s)
-
e^\epsilon\varphi(b(s))b'(s).
\]
Now,
\[
a'(s)=\frac{s+\epsilon}{2\sqrt{2}s^{3/2}},
\qquad
b'(s)=\frac{s-\epsilon}{2\sqrt{2}s^{3/2}}.
\]
Also, $b(s)^2-a(s)^2=2\epsilon.$ Therefore, $\varphi(b(s))
=
e^{-\epsilon}\varphi(a(s)),$
which implies $e^\epsilon\varphi(b(s))=\varphi(a(s)).$ Substituting this into the derivative gives
\[
\frac{d}{ds}g(s;\epsilon)
=
\varphi(a(s))\left(a'(s)-b'(s)\right).
\]
Using the expressions for \(a'(s)\) and \(b'(s)\),
\[
a'(s)-b'(s)
=
\frac{\epsilon}{\sqrt{2}s^{3/2}}.
\]
Hence,
\[
\frac{d}{ds}g(s;\epsilon)
=
\varphi(a(s))
\frac{\epsilon}{\sqrt{2}s^{3/2}}.
\]
Since \(\varphi(a(s))>0\), \(\epsilon>0\), and \(s>0\), we have
\[
\frac{d}{ds}g(s;\epsilon)>0.
\]
Thus, \(g(s;\epsilon)\) is strictly monotonically increasing in \(s>0\)
for every fixed \(\epsilon>0\).
\end{proof}
%\noindent\rule{\linewidth}{0.6pt}

\section{Proofs of the Basis-Design Results}
\label{app:basis_design_analysis}

\subsection{Proof of Lemma~\ref{lemma:coalition_invertibility}}
\label{app:coalition_matrix_structure}

\begin{proof}
Recall that the noise-basis functions are parameterized as $\psi_r(z)
=
q(z)\xi_r(z),
\ r\in[t],$
where $\xi_r(z)\in\mathcal P_{t-1}$. Let $\mathbf v(z)
=
\begin{bmatrix}
1 & z & \cdots & z^{t-1}
\end{bmatrix}^{T}.$
If $\mathbf b_r^{T}$ denotes the $r$-th row of
$\mathbf B$, then $\xi_r(z)
=
\mathbf v(z)^{T}\mathbf b_r.$
Consider a colluding worker node set
$\mathcal T=\{i_1,\ldots,i_t\}\subset[N]$.
For $\ell,r\in[t]$,
\begin{align}
[\mathbf\Psi_{\mathcal T}(\mathbf B)]_{\ell,r}
&=
\psi_r(\alpha_{i_\ell})
\notag\\
&=
q(\alpha_{i_\ell})
\mathbf v(\alpha_{i_\ell})^{T}\mathbf b_r.
\notag
\end{align}
Therefore,
\begin{equation}
\mathbf\Psi_{\mathcal T}(\mathbf B)
=
\mathbf H_{\mathcal T}\mathbf B^{T},
\label{eq:app_coalition_factorization}
\end{equation}
where $\mathbf H_{\mathcal T}
=
\mathbf D_{\mathcal T}\mathbf V_{\mathcal T},$
with $\mathbf D_{\mathcal T}
=
\operatorname{diag}
\left(
q(\alpha_{i_1}),
\ldots,
q(\alpha_{i_t})
\right)$
and
\begin{equation}
\mathbf V_{\mathcal T}
=
\begin{bmatrix}
1 & \alpha_{i_1} & \cdots & \alpha_{i_1}^{t-1}\\
1 & \alpha_{i_2} & \cdots & \alpha_{i_2}^{t-1}\\
\vdots & \vdots & \ddots & \vdots\\
1 & \alpha_{i_t} & \cdots & \alpha_{i_t}^{t-1}
\end{bmatrix}.
\notag
\end{equation}

Since the worker evaluation points are pairwise distinct,
$\mathbf V_{\mathcal T}$ is a nonsingular Vandermonde matrix.
Furthermore, by assumption, $q(\alpha_i)\neq0,
\ i\in[N],$
and hence $\mathbf D_{\mathcal T}$ is nonsingular. Therefore,
$\mathbf H_{\mathcal T}$ is nonsingular for every
$\mathcal T\subset[N]$ with $|\mathcal T|=t$.

Using \eqref{eq:app_coalition_factorization}, $\operatorname{rank}
\bigl(
\mathbf\Psi_{\mathcal T}(\mathbf B)
\bigr) = \operatorname{rank}\bigl(\mathbf H_{\mathcal T}\mathbf B^{T}\bigr)=\operatorname{rank}(\mathbf B^{T})=
\operatorname{rank}(\mathbf B).$
Hence,
$\mathbf\Psi_{\mathcal T}(\mathbf B)$ is nonsingular if and
only if $\mathbf B$ is nonsingular. Moreover, using the standard singular-value inequality for
products of square matrices, $\sigma_{\min}
\bigl(
\mathbf\Psi_{\mathcal T}(\mathbf B)
\bigr) =
\sigma_{\min}
\bigl(
\mathbf H_{\mathcal T}\mathbf B^{T}
\bigr)\geq
\sigma_{\min}(\mathbf H_{\mathcal T})
\sigma_{\min}(\mathbf B^{T})=
\sigma_{\min}(\mathbf H_{\mathcal T})
\sigma_{\min}(\mathbf B).$
Hence, if
$\sigma_{\min}(\mathbf B)\geq\eta>0$, then $\sigma_{\min}
\bigl(
\mathbf\Psi_{\mathcal T}(\mathbf B)
\bigr)
\geq
\eta\,
\sigma_{\min}(\mathbf H_{\mathcal T})
>0,$
where the strict inequality follows from the nonsingularity of
$\mathbf H_{\mathcal T}$. This proves the result.
\end{proof}

\subsection{Proof of Lemma~\ref{lemma:scale_invariance}}
\label{app:scale_invariance}

\begin{proof}
Let $c\in\mathbb C\setminus\{0\}$ and define $\mathbf B_c
=
c\mathbf B.$
Since the noise-basis functions depend linearly on the coefficients of $\mathbf B$, $\psi_r(z;\mathbf B_c)
=
c\,\psi_r(z;\mathbf B),
\qquad r\in[t].
$
Consequently, for every coalition
$\mathcal T\subset[N]$ with $|\mathcal T|=t$,
\begin{equation}
\mathbf\Psi_{\mathcal T}(\mathbf B_c)
=
c\,\mathbf\Psi_{\mathcal T}(\mathbf B).
\label{eq:app_Psi_scaling}
\end{equation}

It follows that
\begin{align}
&\mathbf\Psi_{\mathcal T}(\mathbf B_c)
\mathbf\Psi_{\mathcal T}(\mathbf B_c)^{H}
\notag\\
&\qquad=
|c|^2
\mathbf\Psi_{\mathcal T}(\mathbf B)
\mathbf\Psi_{\mathcal T}(\mathbf B)^{H},
\notag
\end{align}
and therefore
\begin{align}
&
\Bigl(
\mathbf\Psi_{\mathcal T}(\mathbf B_c)
\mathbf\Psi_{\mathcal T}(\mathbf B_c)^{H}
\Bigr)^{-1}
\notag\\
&\qquad=
\frac{1}{|c|^2}
\Bigl(
\mathbf\Psi_{\mathcal T}(\mathbf B)
\mathbf\Psi_{\mathcal T}(\mathbf B)^{H}
\Bigr)^{-1}.
\notag
\end{align}

Since $\mathbf\Phi_{\mathcal T}(\mathbf A)$ is independent of
$\mathbf B$, it follows from the definition of
$\mathbf G_{\mathcal T}(\mathbf A,\mathbf B)$ that
\[
\mathbf G_{\mathcal T}(\mathbf A,\mathbf B_c)
=
\frac{1}{|c|^2}
\mathbf G_{\mathcal T}(\mathbf A,\mathbf B).
\]
Taking the maximum over all coalitions and diagonal entries gives
\begin{equation}
\Gamma_t(\mathbf A,\mathbf B_c)
=
\frac{1}{|c|^2}
\Gamma_t(\mathbf A,\mathbf B).
\label{eq:app_Gamma_scaling}
\end{equation}

Furthermore, from \eqref{eq:app_psi_scaling},
\begin{equation}
\sum_{r=1}^{t}
\left|
\psi_r(\alpha_i;\mathbf B_c)
\right|
=
|c|
\sum_{r=1}^{t}
\left|
\psi_r(\alpha_i;\mathbf B)
\right|.
\label{eq:app_noise_basis_sum_scaling}
\end{equation}

Using \eqref{eq:app_Gamma_scaling} and
\eqref{eq:app_noise_basis_sum_scaling}, the noise-basis term
appearing in the privacy-calibrated objective satisfies
\begin{align}
&
2r\theta
\sqrt{
\frac{\Gamma_t(\mathbf A,\mathbf B_c)}
{s^\star}
}
\sum_{r=1}^{t}
|\psi_r(\alpha_i;\mathbf B_c)|
\notag\\
&\quad=
2r\theta
\sqrt{
\frac{\Gamma_t(\mathbf A,\mathbf B)}
{|c|^2s^\star}
}
|c|
\sum_{r=1}^{t}
|\psi_r(\alpha_i;\mathbf B)|
\notag\\
&\quad=
2r\theta
\sqrt{
\frac{\Gamma_t(\mathbf A,\mathbf B)}
{s^\star}
}
\sum_{r=1}^{t}
|\psi_r(\alpha_i;\mathbf B)|.
\notag
\end{align}
The data-basis contribution to $\rho(\mathbf A,\mathbf B)$ is independent of $\mathbf B$. Hence the worker-wise objective is unchanged for every $i\in[N]$, and therefore
\begin{equation} 
\rho(\mathbf A,c\mathbf B)
= \rho(\mathbf A,\mathbf B).
\label{eq:app_rho_scale_invariance}
\end{equation}
It remains to establish the normalization claim. Let $(\mathbf A,\mathbf B)$ be any admissible design with $\mathbf B$ nonsingular. Then $\|\mathbf B\|_F>0$, and define $\widetilde{\mathbf B} = \frac{\mathbf B}{\|\mathbf B\|_F}.$ Clearly, $\|\widetilde{\mathbf B}\|_F=1.$
Moreover, $\widetilde{\mathbf B}$ remains nonsingular since it is obtained from $\mathbf B$ through multiplication by the nonzero scalar $1/\|\mathbf B\|_F$. Applying \eqref{eq:app_rho_scale_invariance} with $c=1/\|\mathbf B\|_F$ yields $\rho(\mathbf A,\widetilde{\mathbf B}) = \rho(\mathbf A,\mathbf B).$
Thus, every objective value attainable with a nonsingular $\mathbf B$ is also attainable with a normalized nonsingular matrix satisfying $\|\mathbf B\|_F=1$. Hence restricting the scale free basis design problem to normalized noise basis matrices does not alter its attainable objective values.
\end{proof}

\section*{AI Statement}
Generative AI tools were used carefully to assist with literature exploration, language editing, improving clarity and readability, and refining the presentation of the manuscript.

\end{document}